\documentclass[12pt]{article}
\usepackage{amsmath}
\usepackage{amssymb}
\usepackage{amsthm,amsxtra}
\usepackage{epsf}
\usepackage{eepic}
\usepackage{graphicx}
\usepackage{enumitem}
\usepackage{subcaption}
\usepackage{physics}
\usepackage[hidelinks=true]{hyperref}
\hypersetup{
  colorlinks   = true, 
  urlcolor     = blue, 
  linkcolor    = blue, 
  citecolor    = red 
}

\newlength{\mytopmargin}
\newlength{\myleftmargin}
\usepackage[square,sort,comma,numbers]{natbib}
\providecommand{\keywords}[1]
{
  \small	
  \textbf{{Keywords:}} #1
}
\providecommand{\subjclass}[1]
{
  \small	
  \textbf{{MSC 2010:}} #1
}

\newcommand{\half}{\tfrac{1}{2}}
\newcommand{\quarter}{\tfrac{1}{4}}
\newcommand{\tC}{\tilde{t}}
\def\mean#1{\left \langle #1 \right \rangle}
 \def\gtapprox{\buildrel{\lower.7ex\hbox{$>$}}\over {\lower.7ex\hbox{$\sim$}}}

\renewcommand{\theequation}{\thesection.\arabic{equation}}
\newtheorem{lemma}{Lemma}
\newtheorem{prop}[lemma]{Proposition}

\numberwithin{equation}{section}
\newcommand\numberthis{\addtocounter{equation}{1}\tag{\theequation}}

\begin{document}

\begin{center}
{\bfseries\Large Leading corrections to the scaling function on the \\[0.4ex] 
 diagonal for the  two-dimensional Ising model}
\\[2\baselineskip]
{\large P. J. Forrester\footnote{pjforr@unimelb.edu.au}}, %
\\[.5\baselineskip]
{\itshape ARC Centre of Excellence for Mathematical and Statistical Frontiers,\\
School of Mathematics and Statistics, The University of Melbourne, Victoria 3010, Australia.}
\\[\baselineskip]
{\large J. H. H. Perk\footnote{perk@okstate.edu}}
\\[.5\baselineskip]
{\itshape Department of Physics, 145 Physical Sciences, Oklahoma State University, Stillwater, Oklahoma 74078-3072, USA}
\\[\baselineskip]
{A. K. Trinh\footnote{a.trinh4@student.unimelb.edu.au}}
\\[.5\baselineskip]
{\itshape ARC Centre of Excellence for Mathematical and Statistical Frontiers,\\
School of Mathematics and Statistics, The University of Melbourne, Victoria 3010, Australia.}
\\[\baselineskip]
{\large N. S. Witte\footnote{n.s.witte@massey.ac.nz}}
\\[.5\baselineskip]
{\itshape Institute of Fundamental Sciences, Massey University,\\
Private Bag 11222, Palmerston North 4410, New Zealand}
\end{center}

\date{\today}


\begin{abstract}
In the neighbourhood of the critical point, the correlation length of the spin-spin correlation function 
of the two-dimensional Ising model diverges.
The correlation function permits a scaling limit in which the separation $N$ between spins goes to infinity, but the
scaling variable $s =  N(1-t)/2$ remains fixed, where $t$ is the coupling, and $t=1$ the
critical point. Previous work has specified these scaling functions (there is one for the critical point
being approached from above, and another if approached from below) in terms of 
transcendents defined by a particular $\sigma$-form of the degenerate Painlev\'e V equation. 
For the diagonal-diagonal correlation, we characterise the first two leading large
$N$ correction terms to the scaling functions --- these occur at orders $N^{-1}$ and $N^{-2}$ --- 
in terms of solutions of a second order linear differential equation with coefficients given in terms of these transcendents, 
and show how they can be computed. 
We show that the order $N^{-1}$ is trivial and can be eliminated through appropriate variables so that the leading non-trivial correction is of order $N^{-2}$.  
In this respect our result gives precise and full characterisation of claims made in the earlier literature.
\end{abstract}

\noindent
\newline

\section{Introduction}

The two dimensional Ising model holds a central place in mathematical and theoretical physics as a microscopic model of a ferromagnetic
phase transition that allows for exact mathematical analysis. Generally, the theory of statistical mechanics allows for
macroscopic properties of physical systems to be computed from knowledge of the microscopic interactions. In the
case of the ferromagnetic Ising model on a square lattice, where on each lattice site there is an up or down (classical) spin, the
interactions are between nearest neighbours and favour the alignment of spins. In a famous calculation which dates back
to the 1940's Onsager \cite{On44} derived an exact formula for the free energy in the thermodynamic limit. 
It exhibits a critical point, which shows itself as a singularity as a function of the dimensionless coupling.
In addition Onsager and Kaufman found an exact formula for the spontaneous magnetisation in terms of this coupling. 
In the early 1950s C.N.~Yang \cite{Ya52} gave the derivation for this result, exhibiting the experimentally observable
\cite{BWVRMP95} $1/8$-th power singularity as the critical temperature is approached from above. 
The general anisotropic case of the magnetisation was given by Onsager \cite{On49}, and later derived by Chang \cite{Cha_1952} and Potts \cite{Pot_1952}.
A Toeplitz structure for the two-point correlations appears in Eq. (45) of Kaufman and Onsager's 1949 work \cite{KO_1949} on the short-range order in the planar Ising model.
Subsequently Montroll, Potts and Ward \cite{MPW63} gave a Toeplitz determinant formula for the two-point correlation function in the 1960's. 
In the case of the diagonal-diagonal two-point correlation, a Toeplitz determinant formula was already known to Onsager, but its derivation was not published;
see the historical accounts in \cite{Ba11, Ba12, DIK13}.
It is the latter which forms the starting point of the present study.

The Toeplitz determinant formula for the diagonal-diagonal two-point correlation of the two-dimensional Ising model on
a square lattice relates to an $N \times N$ matrix where $N$ is the number of lattice sites that separate the spins.
Above and below the critical point the truncated correlation function decays exponentially fast. But the correlation
length diverges as the critical point is approached, and this in turn leads to the notion of a scaling limit, in which $N$ goes
to infinity, while the product $N(1-t)$ is fixed, with $t$ the scaled coupling such that
$t \to 1$ corresponding to the critical point.  Already established
mathematical results tell us that the Toeplitz determinant can be expanded as an infinite series known as the form
factor expansion \cite{Wu66,WMTB76,Bu01,BL04,LM07,BHMMZ07,WF12}.
Each term in this expansion is a multiple integral of increasing dimension. The scaling limit of this
can be taken term by term, giving the form factor expansion of the scaling limit. On the other hand, another already
established mathematical result is that the Toeplitz determinant can be characterised in terms of the solution of a
Painlev\'e VI non-linear differential equation in so called sigma form --- a result due to Jimbo and Miwa
\cite{JM80, JM81}, with the latter reducing to the degenerate Painlev\'e V equation in sigma form in the scaling limit.

The scaling regime has been investigated by numerous authors, and in particular we want to focus on the influential and pioneering work of Wu et al \cite{WMTB76},
referred to hereafter as WMTB.
Their work treated the non-diagonal correlations on the anisotropic lattice so our comparison will just be for a specialisation of this.
Their primary result concerned the explicit evaluation of the zeroth order term in the scaling regime with a solution of the third Painlev\'e equation (equivalent to the degenerate PV), 
but they did make some observations and claims about the nature of the next correction to this. 
In the following discussion we are going examine their claims and will assume an anisotropic lattice as well. 
The following notations are standard, and common to WMTB; for the generating function variables of the partition function 
\begin{equation}
   z_1 = \tanh \beta E_1, \quad z_2 = \tanh \beta E_2 ,
\end{equation}
in terms of the couplings between neighbouring spins $ K_j=\beta E_j, j=1,2 $ in the horizontal and vertical directions;
and the coefficients of the dispersion relation on the square lattice
\begin{equation}
   a = (1+z_1^2)(1+z_2^2), \quad \gamma_1 = 2z_2(1-z_1^2), \quad \gamma_2 = 2z_1(1-z_2^2) .
\end{equation}
The ferromagnetic critical point is given by either of the relations
\begin{equation}
   z_{1C}z_{2C}+z_{1C}+z_{2C}-1 = 0, \quad a_{C}-\gamma_{1C}-\gamma_{2C} = 0 .
\end{equation}
The deviation of the inverse temperature from the critical value is measured by $ \Delta \beta := \beta-\beta_{C} $.
The symmetrised spatial separation variable for the correlation $ \langle s_{0,0}s_{M,N} \rangle $ is denoted by $ R $ where
\begin{equation}
  R^2 = \left( \frac{\sinh 2\beta E_1}{\sinh 2\beta E_2} \right)^{1/2}M^2 + \left( \frac{\sinh 2\beta E_2}{\sinh 2\beta E_1} \right)^{1/2}N^2 .
\end{equation}
The independent variable used by WMTB in the critical regime is denoted here by $ \tC $ where
\begin{equation}
  \tC = | z_{1}z_{2}+z_{1}+z_{2}-1|\left[ z_1z_2(1-z_1^2)(1-z_2^2) \right]^{-1/4}R .
\end{equation}
In contrast the PVI $t$-variable, employed throughout our study, has the form $ T<T_C $
\begin{multline}
   t-1 = \frac{1}{16z_1^2z_2^2}
\\ \times
   \left[ 1-z_{1}-z_{2}-z_{1}z_{2} \right]\left[ 1+z_{1}+z_{2}-z_{1}z_{2} \right]\left[ 1-z_{1}+z_{2}+z_{1}z_{2} \right]\left[ 1+z_{1}-z_{2}+z_{1}z_{2} \right] .
\end{multline}
Performing expansions as $ \Delta\beta\to 0 $ we note that the WMTB variable and our $ s:=N(1-t)/2 $ are related by
\begin{equation}
   \tC = |s| + {\rm O}(\Delta\beta^2) .
\end{equation} 

WMTB express the separation of the large distance and scale free dependencies of the pair correlation as a large $ R $ expansion
\begin{equation}
   \langle s_{0,0}s_{M,N} \rangle = R^{-1/4} F_{\pm}(\tC) + R^{-5/4} F_{1\pm}(\tC) + {\rm o}(R^{-5/4}) .
\end{equation}
In the summary section WMTB make the claim, see Eq. (2.24), that
\begin{equation}
   \frac{F_{1+}(\tC)}{F_{+}(\tC)} = - \frac{F_{1-}(\tC)}{F_{-}(\tC)} = -\tC R_1
\end{equation} 
where $ R_1 $ is given as an algebraic expression of $ z_{1C}, z_{2C}, E_1, E_2 $, which is independent of $ \tC $.
Furthermore this has consequences for the magnetic susceptibility, which has the expansion about the critical point
\begin{equation}
  \beta^{-1} \chi(T) = C_{0\pm} | 1-T_{C}/T |^{-7/4} + C_{1\pm} | 1-T_{C}/T |^{-3/4} + {\rm O}(1) ,
\end{equation}
where the ratio of the sub-leading to leading coefficients is
\begin{equation}
   \frac{C_{1+}}{C_{0+}} = -\frac{C_{1-}}{C_{0-}} = -\beta_C R_0 ,
\end{equation}
with $ R_0 $ given as another expression of $ z_{1C}, z_{2C}, E_1, E_2 $.

The question addressed in the present work is to characterise the leading corrections to the scaling limit as the solutions of a differential equation. 
Such a question was first raised by WMTB, in the more general context of the two-point correlation in general position.
It is found in Proposition \ref{P1} that the leading corrections --- which appear at orders $N^{-1}$ and $N^{-2}$ --- can be characterised as solutions of linear second order
differential equations, which have as their coefficients the Painlev\'e transcendents characterising the scaling limiting form itself.
Equivalently, the leading corrections to the scaling limit can be characterised as coupled differential systems, involving
both a particular Painlev\'e V equation in sigma form, and a second order linear differential equation. 
This structure has been seen in a number of other recent studies involving Painlev\'e transcendents characterising 
finite size corrections \cite{FM15, BFM17, FT17}.

The characterisation becomes unique once boundary conditions for the equation are specified. 
This task is carried out for both the small and large values of the scaling variable $s$ using the known series expansions about $t=1$ (\S \ref{S3.3})
and a form factor expansion (\S \ref{S3.2}) respectively.
The computation of the small $s$ expansion using known series expansions is particularly interesting. 
It requires a double scaling with the introduction of the scaling variable.
As a result, the explicit form of the series expansion solution of the coupled differential equations can be determined up to arbitrary order, subject to the capacity of the computer
algebra system used for the purpose (\S \ref{s3.4} and Appendix A). 
Using the series form, accurate numerical values of the scaling function and its first two corrections can be made up to sufficiently large values of $s$ that they can be joined up with the large $s$ asymptotic form, without the need to actually make use of the differential equations in this regime. This is carried out in \S \ref{S4}.
For both the scaling function and its leading corrections accurate numerical values can be computed simply by the extrapolation of values from the Toeplitz determinant for a sequence of
values of $N$, allowing for a numerical validation of our analytic results.

\section{Preliminaries}
\subsection{Some definitions and Onsager's Toeplitz formula}
To specify the Ising model in two dimensions, we start with a square lattice of size $(2M+1) \times (2M+1)$,
centred at the origin so that that nodes $(i,j)$ are pairs of integers with $-M \le i,j \le M$. On each node of the
square lattice, there is an associated spin $s_{i,j} \in \{-1,1\}$. The spins interact with their nearest neighbours
in the horizontal and vertical directions according to the dimensionless interaction energy
$$
\beta \mathcal E = - K_1 \sum_{i=-M}^{M-1} \sum_{j=-M}^M s_{i,j} s_{i+1,j} -
K_2 \sum_{i=-M}^{M} \sum_{j=-M}^{M-1} s_{i,j} s_{i,j+1} .
$$
Our interest is in the ferromagnetic case $K_1, K_2 > 0$ which for low temperatures favours neighbouring spins to align.

The corresponding partition function is
$$
Z_{2M+1} = \sum_{\{ s_{i,j} \}} \exp(-\beta  \mathcal E).
$$
This occurs in the normalisation of the formula for the probability $P(\{ s_{i,j} \}_{i,j=-M}^M)$
of a particular configuration
$\{ s_{i,j} \}_{i,j=-M}^M$,
$$
P(\{ s_{i,j} \}_{i,j=-M}^M) = { e^{-\beta \mathcal E} \over Z_{2M+1}}.
$$
The probability in turn occurs in the formula for the infinite lattice spontaneous magnetisation
\begin{equation}\label{2.1}
\mathcal M = \langle s_{0,0} \rangle =
\lim_{M \to \infty}  \sum_{\{ s_{i,j} \}}  s_{0,0} P(\{ s_{i,j} \}_{i,j=-M}^M).
\end{equation}

There is some subtlety in relation to (\ref{2.1}). With
\begin{equation}\label{2.2}
k = \sinh 2 K_1 \sinh 2 K_2 ,
\end{equation}
there is phase transition at $k=1$ separating a high temperature phase $0 < k < 1$, with zero
spontaneous magnetisation, from a low temperature phase $k > 1$ for which this order parameter is non-zero. 
The subtlety is that whether this is a positive value, or negative value, depends on the boundary condition: we choose all spins on the boundary to be pointing up, 
so that the limiting value will be positive. 
With this convention \cite{Ya52} 
\begin{equation}\label{2.3}
\mathcal M = \left \{ \begin{array}{ll} (1 - k^{-2})^{1/8}, & k > 1 \\
0, & 0 < k < 1 \end{array} \right. .
\end{equation}

The spin-spin correlation, between the spin $s_{0,0}$ at the origin, and the spin $s_{m,n}$ at
lattice site $(m,n)$, is for the infinite lattice defined
\begin{equation}\label{2.4}
 \langle s_{0,0}  s_{m,n} \rangle =
\lim_{M \to \infty}  \sum_{\{ s_{i,j} \}}  s_{0,0} s_{m,n} P(\{ s_{i,j} \}_{i,j=-M}^M),
\end{equation}
again with the convention that all spins on the boundary are to be pointing up.
According to Onsager  \cite{Ba11, Ba12,  DIK13}, in the diagonal case $m=n(=N)$
\begin{equation}\label{2.4a}
 \langle s_{0,0}  s_{N,N} \rangle = \det [ a_{i-j} ]_{1 \le i,j\le N},
 \end{equation}
 where the elements are given as the Fourier coefficients 
\begin{equation}\label{2.5} 
a_n = {1 \over 2 \pi} \int_{-\pi}^\pi a(e^{i \theta}) e^{-i n \theta} \, d \theta,
\end{equation}
 with the weight (here $k$ is given by (\ref{2.2}))
\begin{equation}\label{2.6} 
a(\zeta) = \left[ \frac{1 - k^{-1} \zeta^{-1}}{1 - k^{-1} \zeta} \right]^{1/2}, \quad \zeta:=e^{i \theta} .
\end{equation}

From \cite{Wi07} we know that as a consequence of the integral representation of the
${}_2 F_1$ hypergeometric function
$$
{}_2 F_1 (a,b;c;z) = {\Gamma(c) \over \Gamma(c-b) \Gamma(b)}
\int_0^1 x^{b-1} (1 - x)^{c-b-1} (1 - x z)^{-a} \, dx \qquad ({\rm Re} \, c > {\rm Re} \, b > 0),
$$
one has for $1 < k < \infty$
\begin{align}
a_n & = {\Gamma(n+1/2) \Gamma(1/2) \over \pi \Gamma(n+1)} k^{-n} \,{}_2 F_1(-1/2, n+1/2; n+1; k^{-2} ) , \nonumber \\
a_{-n} & = -{\Gamma(n-1/2) \Gamma(3/2) \over \pi \Gamma(n+1)} k^{-n} \,{}_2 F_1(1/2, n -1/2; n+1; k^{-2} ) , \label{F1}
\end{align}
while for $0 \le k < 1$,
\begin{align}
a_n & = {\Gamma(n+1/2) \Gamma(3/2) \over \pi \Gamma(n+2)} k^{n+1} \,{}_2 F_1(1/2, n+1/2; n+2; k^{2} ) , \nonumber \\
a_{-n} & =  -{\Gamma(n-1/2) \Gamma(1/2) \over \pi \Gamma(n)} k^{n-1} \,{}_2 F_1(-1/2, n-1/2; n; k^{2} ). \label{F2}
\end{align}
The formulas for $a_n$ hold for $n \ge 0$, and those for $a_n$ hold for $n \ge 1$.

\subsection{Form factor expansion}
Introduce the variable
\begin{equation}\label{2.7}
t = \left \{ \begin{array}{ll} k^{-2}, & k > 1 \\
k^2, & 0 < k < 1 \end{array} \right. ,
\end{equation}
which is always $ 0<t<1 $ and furthermore write
\begin{equation}\label{2.7a}
C_{N,N}(t) = \langle s_{0,0}  s_{N,N} \rangle.
\end{equation}
The so-called form factor expansions \cite{Wu66,WMTB76,Bu01,BL04,LM07,BHMMZ07,WF12}
are infinite sums involving multiple integrals of increasing dimension, expressing
$ \langle s_{0,0}  s_{N,N} \rangle$ in a form giving immediate information relating to the small
$t$ power series. These read
\begin{equation}\label{1.4.9}
\mean{s_{0,0}s_{N,N}} = (1-t)^{\quarter} \Bigg(1+\sum_{p=1}^{\infty} f^{(2p)}_{N,N} \Bigg),
\quad 
\mean{s_{0,0}s_{N,N}} = (1-t)^{\quarter} \sum_{p=0}^{\infty} f^{(2p+1)}_{N,N},
\end{equation}
for $T<T_C$ and $T>T_C$ respectively. Here
\begin{align*} \label{eq:lowformfac}
f^{(2p)}_{N,N} = \frac{t^{p(N+p)}}{(p!)^2 \pi^{2p}} \int_0^1 dx_1 \cdots \int_0^1 dx_{2p}
&\prod_{k=1}^{2p} x^N_k \prod_{j=1}^p \Bigg[ \frac{(1-tx_{2j})(x^{-1}_{2j}-1)}{(1-tx_{2j-1})(x^{-1}_{2j-1}-1)} \Bigg]^{\half} \\
& \times \prod_{j=1}^{p} \prod_{k=1}^{p} (1-tx_{2k-1}x_{2j})^{-2} \\ & \times \prod_{1 \leq j<k \leq p} (x_{2j-1}-x_{2k-1})^2(x_{2j}-x_{2k})^2 ,
\numberthis
\end{align*} 
and
\begin{align*} \label{eq:highformfac}
f^{(2p+1)}_{N,N} = &\frac{t^{N(p+1/2)+p(p+1)}}{p!(p+1)!\pi^{2p+1}} \int_0^1 dx_1 \cdots \int_0^1 dx_{2p+1} 
\\
& \times \prod_{k=1}^{2p+1} x^N_k \prod_{j=1}^{p+1} \frac{1}{x_{2j-1}}(1-tx_{2j-1})^{-\half}(x^{-1}_{2j-1}-1)^{-\half}
\\
& \times \prod_{j=1}^p x_{2j}(1-tx_{2j})^{\half}(x^{-1}_{2j}-1)^{\half}  \prod_{j=1}^{p+1} \prod_{k=1}^p (1-tx_{2j-1}x_{2k})^{-2}
\\
& \times \prod_{1 \leq j<k \leq p+1} (x_{2j-1}-x_{2k-1})^2
\prod_{1 \leq j<k \leq p} (x_{2j}-x_{2k})^2 .
\numberthis
\end{align*}

For the implied small $t$ expansions we have \cite{Gho_2005}, \cite{FW_2005}, \cite{PA_2009}, \cite{WF12}
\begin{equation}\label{2.8}
C_{N,N}(t)  = 
 \left \{ \begin{array}{ll} \displaystyle
(1 - t)^{1/4} +
{(1/2)_N (3/2)_N \over 4 ((N+1)!)^2} t^{N+1} \Big ( 1 + {\rm O}(t) \Big ), &  k > 1 \\
 \displaystyle 
 {1 \over \sqrt{\pi}} {\Gamma(N+1/2) \over \Gamma(N+1)} t^{N/2}  \Big ( 1 + {\rm O}(t) \Big ), & k < 1.
\end{array} \right.
\end{equation}

\subsection{$\sigma$ function and Painlev\'e VI}
Define the $\sigma$-function
\begin{equation}\label{2.8a}
\sigma(t;N) =  
    \begin{cases}
     \displaystyle t(t-1) \dv{t} \log \langle s_{0,0} s_{N,N} \rangle - {1 \over 4} t, & k > 1 \\[2ex]
     \displaystyle t(t-1) \dv{t} \log \langle s_{0,0} s_{N,N} \rangle  - {1 \over 4}, & 0 < k < 1
    \end{cases} ,
\end{equation}
where $t$ is specified in terms of $k$ according to (\ref{2.7}). It was shown by Jimbo and Miwa \cite{JM80, JM81}
(see \cite{FW04} for a different derivation) that $\sigma(t;N)$ satisfies the particular $\sigma$-form of Painlev\'e VI
(for an account of the latter, see e.g. \cite[\S 8.2]{Fo10})
\begin{align*} 
\bigg[ t(t-1)\dv[2]{\sigma}{t} \bigg]^2 = N^2 &\bigg[ (t-1)\dv{\sigma}{t} - \sigma \bigg]^2 
\\
&- 4\dv{\sigma}{t}\bigg[(t-1)\dv{\sigma}{t}-\sigma-\quarter\bigg] \bigg[t\dv{\sigma}{t}-\sigma\bigg].
\numberthis
\label{eq:pvi}
\end{align*}
We should emphasize that our $t$ is the inverse of Jimbo-Miwa's $t$, that is $t=1/t_{\mathrm{JM}}$.
To be consistent with (\ref{2.8}) we require the boundary conditions
\begin{equation}\label{2.9}
\sigma(t;N) \mathop{\sim}\limits_{t \to 0}
 \begin{cases}
  \displaystyle -\frac{(1/2)_N (3/2)_N}{4(N+1)!N!} t^{N+1} + {\rm O} (t^{N+2}), & k > 1 \\[2ex]             
  \displaystyle - {N \over 2}  - {1 \over 4} + {\rm O} (t) , & 0 < k < 1
 \end{cases} .
\end{equation}

\subsection{The two-point correlation at criticality and its large $N$ expansion}

It is well known that the Toeplitz determinant (\ref{2.4a}) simplifies when $k=1$ \cite{MW73}. The Fourier
coefficients then permit the evaluation $a_n = 1/(\pi (n+1/2))$. Use of the Cauchy double alternant determinant
(see e.g.~\cite[Eq.~(4.33)]{Fo10}) then allows the Toeplitz determinant to be evaluated to give
\begin{equation}\label{2.10}
 \langle s_{0,0}  s_{N,N} \rangle  \Big |_{k=1} =  C_{N,N}(1) =
 \Big ( {2 \over \pi} \Big )^N \prod_{p=1}^{N-1} \Big ( 1 - 
 {1 \over 4 p^2} \Big )^{p-N}.
 \end{equation}
 
Significant for the interpretation of some of our future working is that (\ref{2.10}) has the large $N$ expansion
\cite{MW73}, or Eq. (5) of \cite{AP_1984} for all orders explicitly
\begin{equation}\label{2.11}
 \langle s_{0,0}  s_{N,N} \rangle  \Big |_{k=1}   \mathop{\sim}\limits_{N \to \infty}
 {A \over N^{1/4}} \Big ( 1 - {1 \over 64 N^2} + {\rm O}(N^{-4}) \Big ),
\end{equation}
where, with $\zeta'(z)$ the derivative of the Riemann zeta function,
\begin{equation}\label{2.12}
  A = 2^{1/12} \exp(3 \zeta'(-1)).
\end{equation} 
Such an expansion for the non-diagonal correlations has also been given in \cite{AP_1984}.

\section{Scaling limit about criticality}
\subsection{Scaling variable and scaling function}
With $t$ defined as in (\ref{2.7}), we know that the model is critical for $t=1$. To quantify the meaning of this
in relation to the diagonal two point correlation function, we note from (\ref{2.8}) that for $t < 1$ the latter assumes
its limiting value exponentially fast in $N$. This suggests the introduction of a correlation length $\xi_{\pm}$ by
setting (see e.g.~\cite[Eq.~(10.114)]{Mc12})
\begin{equation}\label{2.13}
  t^N = e^{- \sqrt{2} N / \xi_{\pm}} ,
\end{equation}
(here one should interpret $\sqrt{2} N$ as the distance from the origin of $(N,N)$ on the square lattice). 
As $t \to 1^-$ it then follows that
\begin{equation}\label{corr_length}  
 \xi^{\pm} = \frac{\sqrt{2}}{|\log t|} \mathop{\sim}\limits_{t \to 1^{-}} \frac{\sqrt{2}}{|1-t|} .
\end{equation}
 
Next, for $N$ large,  introduce the scaled lattice position $(N,N) \mapsto (n,n)$ by $n=N/(2 \xi_{\pm})$, and the corresponding
distance to the scaled coordinate $(n,n)$ by $s = \sqrt{2} n$. It then follows from (\ref{corr_length}) that as
$t \to 1^-$ and simultaneously $N \to \infty$ (the scaling limit) that
\begin{equation}\label{scaling-var} 
 s = \frac{N(1-t)}{2} .
\end{equation}
Our subsequent interest is in the form of  the diagonal two point correlation function as a function of $s$ in the scaling limit.

However a significant issue is which scaling variable, $ s $ or $ N/\xi^{\pm} $, better reflects the true behaviour of the correlation
even though either are acceptable and identical at the lowest order.
Firstly let us define the diagonal scaling function $ G^{\pm}(t;N) $ by
\begin{equation}
  \langle s_{0,0}s_{N,N} \rangle =:\;(1-t)^{1/4}\;G^{\pm}(t;N) .
\end{equation}
For present purposes existing results in the variable $s$ are most informative. 
From Eqs. (10) and (11) in \cite{AP_2003} we know the diagonal scaling function exactly up to order ${\rm O}(N^{-2})$:
\begin{align}
\langle s_{0,0}s_{N,N} \rangle
=\;&|1-k^{-2}|^{1/4}\;\left( G^{\pm}(N |\ln k|)+N^{-2}G^{(2)\pm}(N |\ln k|)+{\rm O}(N^{-4}) \right)
\nonumber\\
=\;&|1-t^{\mp 1}|^{1/4}\;
\bigg[G^{\pm}\left(N\left\{\tfrac{1}{2}(1-t)+\tfrac{1}{4}(1-t)^2 +\tfrac{1}{6}(1-t)^3+\ldots\right\}\right)
\nonumber\\
&\mbox{}\hskip50mm
+N^{-2}G^{(2)\pm}\left(N\left\{\tfrac{1}{2}(1-t)+\ldots\right\}\right)+\ldots\bigg]
\nonumber\\
=\;&(1-t)^{1/4}\left[1+(1\pm 1)\frac{s}{4N}+(1\pm 1)\frac{5s^2}{16N^2}+\ldots\right]
\nonumber\\
&\qquad \times\bigg[G^{\pm}\left(s+\frac{s^2}{N}+\frac{4s^3}{3N^2}\ldots\right)+N^{-2}G^{(2)\pm}(s+\ldots)+\ldots\bigg]
\nonumber\\
=\;&(1-t)^{1/4}\left[G_0^{\pm}(s)+N^{-1}G_1^{\pm}(s)+N^{-2}G_2^{\pm}(s)+{\rm O}(N^{-3})\right].
\label{trueGdefn}
\end{align}
Here we used the exact diagonal correlation length $1/|\ln k|=-2/\ln t$. 
The $+$ sign in $ G^{\star \pm}_{\star} $ indicates $ T>T_C $, whereas the $-$ sign refers to $ T<T_C $ and the upper/lower signs match throughout a formula. 
Thus we find
\begin{align}
G_0^{\pm}(s)=\;
& G^{\pm}(s),
\label{G0-Id}\\
G_1^{\pm}(s)=\;
&s^2\frac{\mathrm{d}G^{\pm}(s)}{\mathrm{d}s}+\frac{1\pm1}{4}sG^{\pm}(s),
\label{G1-Id}\\
G_2^{\pm}(s)=\;
&\frac{1}{2}s^4\frac{\mathrm{d}^2G^{\pm}(s)}{\mathrm{d}s^2}+\left[\frac{4}{3}+\frac{1\pm1}{4}\right]s^3\frac{\mathrm{d}G^{\pm}(s)}{\mathrm{d}s}
 +\frac{5}{4}\times\frac{1\pm1}{4}s^2G^{\pm}(s)+G^{(2)\pm}(s).
\label{G2-Id}
\end{align}
Given $G_0^{\pm}(s)$, $G_1^{\pm}(s)$ and $G_2^{\pm}(s)$, we can find $G^{\pm}(s)$ and $G^{(2)\pm}(s)$ using \eqref{G0-Id} and \eqref{G2-Id},
or equivalently, also using \eqref{G1-Id},
\begin{align}
G^{\pm}(s)=\;
&G_0^{\pm}(s),
\label{trueG0}\\
G^{(2)\pm}(s)=\;
&G_2^{\pm}(s)-s^2\frac{\mathrm{d}G_1^{\pm}(s)}{\mathrm{d}s}-\frac{1}{4}sG_1^{\pm}(s)
 +\frac{1}{2}s^4\frac{\mathrm{d}^2G_0^{\pm}(s)}{\mathrm{d}s^2}+\frac{11}{12}s^3\frac{\mathrm{d}G_0^{\pm}(s)}{\mathrm{d}s}.
\label{trueG2}
\end{align}
Also $G^{(1)\pm}(s)\equiv G^{(3)\pm}(s)\equiv0$, and $G_3^{\pm}(s)$ can be expressed in $G^{(0)\pm}(s)\equiv G^{\pm}(s)$ and $G^{(2)\pm}(s)$.

With our aim to characterise the leading correction terms to the scaling function, we thus have the significant property that only the even inverse
powers in $N$ are independent.
Another clue pointing towards this conclusion is that the PVI sigma form only has $N^2$ as a parameter, and not $N$. 
However the symmetry $N \to -N$ is broken by the boundary conditions (see \eqref{2.9}) which don't have this symmetry, so terms $N^{-odd}$ orders still appear. 
But what this means is that the coefficients of the $N^{-odd}$ terms are trivially related to the $N^{-even}$ coefficients that appear in higher orders.

It is easy to see that Eq.(2.21) in WMTB is also correct to ${\rm O}(N^{-2})$ [they use ${\rm o}(N^{-1})$],
as the only difference is the use of a different correlation length good beyond first order. 
Therefore $\hat F_{1\pm}(t)\equiv 0$ also.
However, expanding $\langle s_{0,0}s_{N,N} \rangle$ as in Eq.(2.22), $F_{\pm}(t)$ differs from $\hat F_{\pm}(t)$
by a factor $|1-k^{-2}|^{1/4}/R^{-1/4}$, leading to Eqs.(2.23)--(2.25) in WMTB. Now
\begin{equation}
F_{1\pm}(t)=\mp tR_1 F_{\pm}(t) + 0\frac{{\rm d}F_{\pm}}{{\rm d}t},
\end{equation}
without the derivative term, as there is no change in the argument $t$.

\subsection{Form factor expansion for large $s$}  \label{S3.2}
It is well known \cite{Mc12} that the form factor expansions \eqref{1.4.9}--\eqref{eq:highformfac}, apart from the factor
$(1 - t)^{1/4}$, permit well defined scaling limits, with each term again expressed as a multiple integral with
increasing number of dimensions, now depending on $s$. In particular, in the scaling limit
\begin{gather}
\begin{aligned}\label{in1}
f_{N,N}^{(1)} & \to {e^{-2s} \over \pi} \int_0^\infty d X_1 \, e^{-2 s X_1}
\left[ {1 \over (1 + X_1) X_1} \right]^{1/2}, \\
f_{N,N}^{(2)} & \to {e^{-2s} \over \pi^2} \int_0^\infty d X_1 \,  \int_0^\infty d X_2 \,  e^{-2 s (X_1+X_2)}
\left[ {(1 + X_2) X_2 \over (1 + X_1) X_1} \right]^{1/2} {1 \over (1 + X_1 + X_2)^2}.
\end{aligned}
\end{gather}
In the original form factor expansions, successive terms contribute at higher order to the
small $t$ expansion; recall (\ref{2.8}). Now one can check that successive terms contribute
at higher order to the {\it large} $s$ expansion.

It is furthermore the case that the integrals in (\ref{in1}) can be evaluated in terms of $K_0$ and $K_1$
Bessel functions. It therefore follows that with \eqref{trueGdefn} we have
\begin{gather}
\begin{aligned}
 \lim_{N \to \infty} G^+(t;N) &= G^+_{0}(s)   \mathop{\sim}\limits_{s \to \infty} {1 \over \pi} K_0(s), \\
 \lim_{N \to \infty} G^-(t;N) &= G^-_{0}(s)  \mathop{\sim}\limits_{s \to \infty} 1 + {1 \over \pi^2}\Big[ s^2 (K_1(s)^2 - K_0(s)^2) - s K_0(s) K_1(s) + {1 \over 2} K_0(s)^2 \Big] .
\end{aligned}
\label{in2}
\end{gather}
 
The derivation of (\ref{in1}) from the integrals in (\ref{1.4.9}) involves the change of variables
$x_j = 1-(1 - t) X_j = 1-2s X_j/N$. This allows the integrals to be also expanded in a
$1/N$ series, with like powers of $1/N$ again having the property that successive terms contribute
at higher order to the large $s$ expansion. Extending (\ref{in2}) we have that for large $N$
\begin{equation}\label{2.15}
 G^{\pm}(t;N) = G^{\pm}_{0}(s) + \frac{1}{N}  G^{\pm}_{1}(s) + \frac{1}{N^2}  G^{\pm}_{2}(s) + {\rm O}(N^{-3}),
\end{equation}
where
\begin{gather}
\begin{aligned}
G^{+}_{1}(s) \mathop{\sim}\limits_{s \to \infty} & \frac{s}{2\pi} \lbrack K_0(s) -2s K_1(s) \rbrack, \\
G^{-}_{1}(s) \mathop{\sim}\limits_{s \to \infty} & -\frac{s^2}{\pi^2}\left[ s \left( K_0(s){}^2-K_1(s){}^2 \right)+K_1(s) K_0(s) \right] ,
\end{aligned}
\label{in3}
\end{gather}
and
\begin{gather}
\begin{aligned}
G^{+}_{2}(s) \mathop{\sim}\limits_{s \to \infty} & \frac{s}{24\pi} \left[ s(12s^2+13) K_0(s)-\left(32 s^2+1\right) K_1(s) \right], \\
G^{-}_{2}(s) \mathop{\sim}\limits_{s \to \infty} & -\frac{s}{24\pi^2} \left[ s(32s^2-1)\left( K_0(s){}^2-K_1(s){}^2 \right)+(20s^2+1) K_1(s) K_0(s) \right] .
\end{aligned}
\label{in4}
\end{gather}
For the second-order correction $ G^{(2)\pm}(s) $ in the large $ s $ regime we have at high temperature
\begin{equation}
  G^{(2)+}(s) \mathop{\sim}\limits_{s \to \infty} -\frac{s}{24\pi} \left[ 2 s K_0(s)+K_1(s) \right] ,
\label{G2-high-T_large-s}
\end{equation}
whereas in the low temperature regime
\begin{equation}
  G^{(2)-}(s) \mathop{\sim}\limits_{s \to \infty} -\frac{s}{24\pi^2}\left[ s(K_1(s){}^2-K_0(s){}^2)+K_1(s) K_0(s) \right] .
\label{G2-low-T_large-s}
\end{equation}
All the relations \eqref{G0-Id}--\eqref{G2-Id}, \eqref{trueG0}--\eqref{trueG2} also apply to these asymptotic formulae,
as \eqref{in3} follows from \eqref{in2} using \eqref{G1-Id} above and applying \eqref{trueG2},
\eqref{in4} becomes indeed \eqref{G2-high-T_large-s} and \eqref{G2-low-T_large-s}. 

In \cite{PA_2009} Perk and Au-Yang gave an expansion for the diagonal correlation in the high temperature regime, of the form
\begin{equation}
  C_{N,N} = \frac{t^{N/2}}{\sqrt{\pi N}(1-t)^{1/4}}\exp\bigg(\sum_{j=1}^m\sum_{s=0}^{\lfloor j/2\rfloor}\frac{p_{j,s}x^{j-2s}}{N^j}\bigg) ,
\end{equation}
truncated at values of $m=1,2,\ldots$, while defining
\begin{equation}
  x=\frac{1+t}{1-t} ,
\end{equation}
with computable coefficients $p_{j,s}$'s.
Similarly, there is an expansion for the dual diagonal correlation or equivalently the low temperature diagonal correlation (corrected from the original)
\begin{equation}
  C_{N,N}^{\ast} = (1-t)^{1/4}+\frac{t^{N+1}}{2\pi N^2 (1-t)^{7/4}}\exp\bigg(\sum_{j=1}^m\sum_{s=0}^{\lfloor j/2\rfloor}\frac{p^{\ast}_{j,s}x^{j-2s}}{N^j}\bigg) ,
\end{equation}
also truncated at values of $m=1,2,\ldots$, with the same $x$, and coefficients $p^{\ast}_{j,s}$.
The explicit results up to $m=10$ are 
{\small
\begin{multline}
C_{N,N} = \frac{k^N}{\sqrt{\pi N}(1-k^2)^{1/4}}
\exp\bigg(
-\frac{x}{8N}+\frac{(x^2-1)}{16N^2}-\frac{x(25x^2-27)}{384N^3}
\\
+\,\frac{(x^2-1)(13x^2-5)}{128N^4}-\frac{x(1073x^4-1830x^2+765)}{5120N^5}
\\
+\,\frac{(x^2-1)(412x^4-425x^2+61)}{768N^6}
\\
-\,\frac{x(375733x^6-886725 x^4+660723x^2-150003)}{229376N^7}
\\
+\,\frac{(x^2-1)(23797x^6-40211x^4+18055x^2-1385)}{4096N^8}
\\
-\,\frac{x(55384775x^8-167281524x^6+179965314x^4-79479684x^2+11415087)}{2359296N^9}
\\
+\,\frac{(x^2-1)(2180461x^8-5127404x^6+3945946x^4-1048244x^2+50521)}{20480N^{10}}
+\cdots\bigg) ,
\label{asyTg}
\end{multline}
}
and
{\small
\begin{multline}
C_{N,N}^{\ast} = (1-k^2)^{1/4}+\frac{k^{2N+2}}{2\pi N^2 (1-k^2)^{7/4}}
\exp\bigg(-\frac{7x}{4N}+\frac{17x^2-10}{8N^2}-\frac{901x^3-783x}{192N^3}
\\
+\,\frac{899x^4-1062x^2+194}{64N^4}-\frac{131411x^5-196770x^3+66375x}{2560N^5}
\\
+\,\frac{83591x^6-151767x^4+75033x^2-6730}{384N^6}
\\
-\,\frac{17052139x^7-36416187x^5+23770797x^3-4402125x}{16384N^7}
\\
+\,\frac{11282939x^8-27723492x^6+22515930x^4-6419700x^2+344834}{2048N^8}
\\
-\,\frac{37620804281x^9-104587369452x^7+101707083486x^5-39418182684x^3+4677930225x}{1179648N^9}
\\
+\,\frac{\displaystyle{2049064082x^{10}-6360721245x^8+7210080180x^6-3544939170x^4+670637250x^2-24119050}}{10240N^{10}}
\\
+\cdots\bigg) .
\label{asyTs}
\end{multline}
}
Because $ x $ diverges in the scaling limit the only way of matching the Perk and Au-Yang expansions with our own is through the large $s$ expansion.
This works because in each term of the above expansion the degree of $x$ in the numerator equals the degree of $N$ in the denominators.
Expanding \eqref{asyTg} and \eqref{asyTs} in a large $N$ expansion up to order $ N^{-2} $ we find that the coefficients of the leading order and 
the $ N^{-1} $, $ N^{-2} $ corrections exactly match the large $s$ asymptotics expansions of \eqref{in2}, \eqref{in3} and \eqref{in4} respectively in both the high and low temperature cases.

\subsection{Small $s$ expansion of $G^\pm(t;N)$} \label{S3.3}
It is known from \cite[Eq. (23)]{MM12} that the particular Painlev\'e VI $\tau$-function possesses a local expansion about $t=1$ such that
\begin{equation}
  C_{N,N}(t) = \langle s_{0,0}s_{N,N} \rangle = \sum^{N}_{p=0}\sum^{\infty}_{n=0} d^{(p,n)}\left( \log|t-1| \right)^{p}(t-1)^{p^2+n} .
\end{equation}
However, working directly from the Toeplitz form (\ref{2.4a}), and using a combination of numerical and analytic reasoning, 
a more refined expansion was given earlier in \cite[Eq.~(4.16), Eq.~(4.17)]{ONGP01}
\begin{multline}
  C_{N,N}(t) = C_{N,N}(1)\, t^{\mp 1/8}
\\
  \times \sum_{p=0}^\infty 4^p\left( \log|\tau|+\frac{1}{2}\psi(N+1)+\frac{1}{2}\psi(N)-\psi(1)-\log 4 \right)^{p}\left( \frac{1}{4} N\tau \right)^{p^2}
\\
  \times \prod^{p-1}_{k=1} (N^{-2}-k^{-2})^{p-k}
\\
  \times \left\{ 1+\frac{1}{8}\left[ 1+2(N^2-p^2) \right]\tau^2 + {\rm O}(\tau^3) \right\}.
\label{t=1_expansion}
\end{multline}
Here $ \tau := \mp\frac{1}{2}(t^{1/4}-t^{-1/4}) $ according to $ T>T_c $ (upper sign) or $ T<T_c $ (lower sign) 
and $ \psi(x) $ is the log-derivative of the Gamma function. The constant factor $C_{N,N}(1)$ has the evaluation (\ref{2.10}).
While \eqref{t=1_expansion} is not exact because it misses the "leading logarithm" caveat, namely with $ \tau^{p^2}\log^{p}|\tau| $ when $ \log^{p}|\tau| $ first appears,
we only require the terms $ a+b\tau\log|\tau|+c\tau $ arising from $ p=0,1 $.
The $ \tau\log|\tau| $ term was first given in \cite{KAP_1986}.
We note that such a general form was proposed in Eq. (2.27) on page 21 of Kong's thesis \cite{Kon_1987}. 
The first few terms close to criticality are in (2.67)-(2.78) there, adding the first five $|T-T_c|\log|T-T_c|$ corrections to the $T=T_c$ result.

The key result we require here is that each term appearing in the last factor of \eqref{t=1_expansion} at 
order $ \tau^q $ has a coefficient which is a polynomial of degree $ q $, i.e. is of order $ {\rm O}(N^q) $ as $ N\to \infty $. 
Following from the introduction of the scaling variable (\ref{scaling-var}),
two simple facts can be deduced. The first is that
\begin{equation}
   \tau = \pm \frac{s}{2N}\left( 1+\frac{s}{N}+\frac{11}{8}\frac{s^2}{N^2} \right) + {\rm O}(s^3) ,
\end{equation}
and the second is
\begin{equation}
   \log|\tau| = \log|s| - \log 2N + \frac{s}{N} + +\frac{7}{8}\frac{s^2}{N^2} + {\rm O}(s^3) .
\end{equation}
The latter equation implies that, upon the introduction of (\ref{scaling-var}), there are no $ \log N $ contributions due to the cancellation 
of $ -\log N $ with the leading order asymptotics of $ \psi(N) $, and therefore the large $ N $ expansion has only algebraic terms
in $ N $.

Assembling all these to compute $ N^{1/4} C_{N,N}(t) $ for small $s$ and large $N$ we note:
\begin{itemize}
\item 
 Only terms with $ p=0 $ and $ p=1 $ are needed to order $ {\rm O}(s^3) $.
\item
 the sub-leading corrections to $ C_{N,N}(1) $ are of order $ {\rm O}(N^{-3}) $, as seen from (\ref{2.11}), so we only require the first two orders.
\item
 the same applies to $ \prod^{p-1}_{k=1} (N^{-2}-k^{-2})^{p-k} $.
\item
 the only terms which remain to order $ {\rm O}(N^{-3}) $ and $ {\rm O}(s^3) $ are those from $ t^{\mp 1/8} $, 
 $ (N\tau)^{p^2} $ and $ (\log|\tau|+\ldots)^p $.
\end{itemize}

Considering these points together, recalling the definitions (\ref{trueGdefn}) and (\ref{2.12})
and using the notation $\gamma_E$ for Euler's constant, we obtain for the small $s$ expansion
\begin{multline}\label{3.13}
(2s)^{1/4}G^\pm(t;N) = A \left\{ 1 \pm \frac{1}{2}s\left[ \log|s|+\gamma_E-\log 8 \right] +\frac{1}{16}s^2 + {\rm O}(s^3) \right\}  
  \\
  + \frac{A}{N} \left\{ \pm \frac{1}{4}s \pm \frac{1}{2}s^2 \pm \frac{4\pm 1}{8}s^2\left[ \log|s|+\gamma_E-\log 8 \right]+ {\rm O}(s^3)  \right\} 
  \\
  + \frac{A}{N^2} \left\{ -\frac{1}{64} \mp \frac{1}{24}s \mp \frac{1}{128}s\left[ \log|s|+\gamma_E-\log 8 \right] + \frac{63 \pm 256}{1024}s^2 + {\rm O}(s^3) \right\} 
  + {\rm O}(N^{-3}) ,
\end{multline}
for $ T>T_c $ or $ T<T_c $.

\subsection{Scaling limit of $\sigma(t;N)$ function}\label{s3.4}
Recalling the definitions (\ref{2.7a}), (\ref{2.8a}), (\ref{trueGdefn}) and (\ref{2.15}) it follows that with
the scaling variable $s$ fixed 
\begin{equation}\label{3.14}
\sigma^{\pm}(t;N) = {\sigma}^{\pm}_0(s) + \frac{1}{N} {\sigma}^{\pm}_1(s) + \frac{1}{N^2} {\sigma}^{\pm}_2(s) + {\rm O}( N^{-3} ),
\end{equation}
where
\begin{align}
{\sigma}^{\pm}_0(s) &= \frac{1}{G_0^{\pm}(s)}s\dv{G_0^{\pm}(s)}{s} \quad \text{for } T\gtrless T_C ,
\label{eq:sigma0G}\\
{\sigma}^{\pm}_1(s) &=
\begin{cases}
 \frac{\displaystyle 1}{\displaystyle G_0^{+}(s)^{2}}
 \left[ \displaystyle -2s^2\dv{G_0^{+}(s)}{s}G_0^{+}(s) + sG_0^{+}(s)\dv{G_1^{+}(s)}{s} - s\dv{G_0^{+}(s)}{s}G_1^{+}(s) -\frac{s}{2} \right], T>T_C \\
 \frac{\displaystyle 1}{\displaystyle G_0^{-}(s)^{2}}
 \left[ \displaystyle -2s^2\dv{G_0^{-}(s)}{s}G_0^{-}(s) + sG_0^{-}(s)\dv{G_1^{-}(s)}{s} - s\dv{G_0^{-}(s)}{s}G_1^{-}(s) \right], T<T_C
\end{cases} .
\label{eq:sigma1G}
\end{align}
The equation (\ref{eq:sigma0G}) can be solved immediately for $G_0^\pm(s)$ to give
\begin{equation}
G_0^{\pm}(s) = \frac{A}{(2s)^{1/4}}\exp \Big ( \int_0^s \frac{{\sigma}_0^{\pm}(x)+\frac{1}{4}}{x} \, dx \Big ).
\label{3.17}
\end{equation}
And making use of (\ref{eq:sigma0G}) in (\ref{eq:sigma1G}) allows the latter to be solved for
$G_1^\pm(s)$,
\begin{gather}
\begin{aligned}
G_1^{-}(s) &= G_0^{-}(s) \int_0^s \Big (\frac{1}{x} {\sigma}^{-}_1(x) + 2 {\sigma}_0^{-}(x)\Big ) \, dx ,
\\
G_1^{+}(s) &= G_0^{+}(s) \left\{ \frac{1}{2}s + \int_0^s \Big (\frac{1}{x} {\sigma}^{+}_1(x) + 2 {\sigma}_0^{+}(x) \Big )\, dx \right\}.
\end{aligned}
\label{3.18}
\end{gather}
Here all constants of integration are chosen to be consistent with (\ref{3.13}).
The relations governing $G^{\pm}_2(s)$ are
\begin{multline}
  G^+_2(s) = G^+_0(s) \left\{ -\frac{1}{64} + \frac{5}{8}s^2 \right.
  \\ \left.
  + \int_0^s \Big( x^{-1}\sigma^+_2(x) + 3x\sigma^+_0(x) + x[\sigma^+_0(x)]^2 + \frac{5}{2}x^2 \sigma^{+'}_0(x) + x^2\sigma^+_0(x)\sigma^{+'}_0(x) \Big)\,dx \right\} ,
\label{highT_G2}
\end{multline}
and
\begin{multline}
  G^-_2(s) = G^-_0(s) \left\{ -\frac{1}{64} \right.
  \\ \left.
   + \int_0^s \Big( x^{-1}\sigma^-_2(x) + 2x\sigma^-_0(x) + x[\sigma^-_0(x)]^2 + 2x^2 \sigma^{-'}_0(x) + x^2\sigma^-_0(x)\sigma^{-'}_0(x) \Big)\,dx \right\}.
\label{lowT_G2}
\end{multline}
Recalling the definition \eqref{trueGdefn} and associated discussion we note the following relations between the 
$ G^{(2)\pm} $ and the $ G^{\pm}_{2} $
\begin{equation}
G^{(2)\pm}(s) = 
\begin{cases}   
\displaystyle G_2^{+}(s) - \frac{1}{2}s^4 \dv[2]{s} G^{+}_0(s) - \frac{11}{6}s^3 \dv{s} G^{+}_0(s) - \frac{5}{8}s^2 G^{+}_0(s), 	& T>T_C \\[2ex]              
\displaystyle G_2^{-}(s) - \frac{1}{2}s^4 \dv[2]{s} G^{-}_0(s) - \frac{4}{3}s^3 \dv{s} G^{-}_0(s),									& T<T_C \\ 
\end{cases} .             
\end{equation}

Moreover, the fact that $\sigma^{\pm}(t;N)$ satisfies the same $\sigma$ Painlev\'e VI equation (\ref{eq:pvi}) in both regimes allows
differential equation characterisations of ${\sigma}^{\pm}_0(s)$, ${\sigma}^{\pm}_1(s)$ and ${\sigma}^{\pm}_2(s)$ to be deduced as a corollary.
We will not indicate the high/low temperature regimes for ease of reading.

\begin{prop}\label{P1}
Introduce the scaling variable (\ref{scaling-var}) and suppose that the solution to ~(\ref{eq:pvi}) can be written in the form (\ref{3.14}). 
Then the leading order function ${\sigma}_0(s)$ satisfies the particular Painlev\'{e} V $\sigma$-form
\begin{align}
(s{\sigma}_0''(s))^2 = 4(s{\sigma}_0'(s) - {\sigma}_0(s))^2 - 4({\sigma}_0'(s))^2(s{\sigma}_0'(s)-{\sigma}_0(s)) + ({\sigma}_0'(s))^2 ,
\label{eq:pV}
\end{align}
and ${\sigma}_1(s)$ satisfies the second order inhomogeneous linear differential equation
\begin{equation}
   A(s){\sigma}_1''(s) + B(s){\sigma}_1'(s) + C(s){\sigma}_1(s) = D_1(s),
\label{Pode-1}
\end{equation}
where
\begin{align}
  A(s) 	& = \frac{1}{2}s^2 {\sigma}_0''(s),  \nonumber \\
  B(s) 	& = s({\sigma}_0'(s))^2 - 2s(s{\sigma}_0'(s) - {\sigma}_0(s)) + 2{\sigma}_0'(s)(s{\sigma}_0'(s) - {\sigma}_0(s) - 1/4),  \nonumber \\
  C(s) 	& = 2(s{\sigma}_0'(s) -{\sigma}_0(s)) - ({\sigma}_0'(s))^2,  \nonumber \\
  D_1(s) 	& = s^3({\sigma}_0''(s))^2 + 2{\sigma}_0'(s)(s{\sigma}_0'(s)-{\sigma}_0(s))(s{\sigma}_0'(s)-{\sigma}_0(s)-1/4).
\label{eq:D(s)}
\end{align}
Furthermore the second correction ${\sigma}_2(s)$ satisfies the second order inhomogeneous linear differential equation
\begin{equation}
   A(s){\sigma}_2''(s) + B(s){\sigma}_2'(s) + C(s){\sigma}_2(s) = D_2(s),
\label{Pode-2}
\end{equation}
where
\begin{multline}
  D_2(s) = 
     (s{\sigma}_0'(s)-{\sigma}_0(s))^2(s{\sigma}_0'(s)-{\sigma}_0(s)-1/4)
   + 2s^3{\sigma}_0'(s){\sigma}_0''(s)(s{\sigma}_0'(s)-{\sigma}_0(s)-1/4)
\\
   + (s^2{\sigma}_0''(s))^2(3+s^2+{\sigma}_0(s)-3s{\sigma}_0'(s))
   + \frac{1}{2}s^5{\sigma}_0''(s){\sigma}_0'''(s)
   - \frac{1}{4}s^6({\sigma}_0'''(s))^2 .
\label{eq:D2(s)}
\end{multline}
\end{prop}

\begin{proof}
We begin with the $\sigma(t;N)$ Painlev\'{e} VI equation (\ref{eq:pvi}) and substitute the proposed form 
$\sigma(t;N) = {\sigma}_0(s) + \frac{1}{N}{\sigma}_1(s) + \frac{1}{N^2}{\sigma}_2(s)$ 
and the scaled variable $s = N(1-t)/2$ to replace $t$. 
Expanding the first, second and third terms respectively gives
\begin{align*}
&[t(t-1)\sigma''(t;N)]^2 = \frac{1}{4}[s {\sigma}''_0(s)]^2 N^2 + [\frac{1}{2}s^2 {\sigma}_0''(s) {\sigma}''_1(s) - s^3{\sigma}''_0(s)^2] N + {\rm O}(1) ,
\\
&\begin{aligned}
N^2 [(t-1)\sigma'_N(t;N) - \sigma(t;N)]^2 =\; &[s{\sigma}'_0(s) - {\sigma}_0(s)]^2 N^2 
\\
& + 2 [s{\sigma}'_0(s) - {\sigma}_0(s)] [s{\sigma}'_1(s) - {\sigma}_1(s)] N + {\rm O}(1) ,
\end{aligned}
\\
&\begin{aligned}
4\sigma'_N(t;N)&[(t-1) \sigma'_N(t;N) - \sigma(t;N) - 1/4] [t\sigma'_N(t;N) - \sigma(t;N)] =
\\
&   {\sigma}'_0(s)^2 [s{\sigma}'_0(s) - {\sigma}_0(s) - 1/4] N^2 + \Big( 2 {\sigma}'_0(s){\sigma}'_1(s)[s{\sigma}'_0(s)-{\sigma}_0-1/4]
\\
& + {\sigma}'_0(s)^2[s{\sigma}'_1(s) - {\sigma}_1(s)] - 2{\sigma}'_0(s)[s{\sigma}'_0(s)-{\sigma}_0(s)][s{\sigma}'_0(s)-{\sigma}_0(s)-1/4] \Big) N 
\\
&+ {\rm O}(1) .
\end{aligned}
\end{align*}
Comparing the coefficients of ${\rm O}(N^2)$ and ${\rm O}(N)$ produces (\ref{eq:pV}) and  (\ref{Pode-1}) respectively.
Eq. \eqref{Pode-2} with \eqref{eq:D2(s)} follows from taking the above working to the next order.
\end{proof}

As a consequence of \eqref{G1-Id} we observe
\begin{equation}
  {\sigma}_{1}^{\pm} = s^2\frac{d}{ds}{\sigma}_{0}^{\pm} - s{\sigma}_{0}^{\pm} .
\end{equation}
We now directly verify the claim that $ {\sigma}_{1}^{\pm} $ as given above is a solution to Eq. \eqref{Pode-1}. 
If given some functions $y_0(s)$, $y_1(s)$ such that $y_1=s^2y_0'-sy_0$ and let
\begin{align*}
	P_0(s)&=(sy_0'')^2-4(sy_0'-y_0)^2+4(y_0')^2(sy_0'-y_0)-(y_0')^2
\\
	P_1(s)&=A(s)y_1''+B(s)y_1'+C(s)y_1-D_1(s) ,
\end{align*}
(i.e. the solutions to the equations $P_0(s)=0$ and $P_1(s)=0$ are ${\sigma}_0$ and ${\sigma}_1$), it is easy to confirm the relation 
\begin{align}
	P_1(s)=\frac{1}{4}s^2\frac{d}{ds}P_0(s) .
\end{align}

We remark that the characterisation of ${\sigma}_0(s)$ by the non-linear equation (\ref{eq:pV}) was
first obtained by Jimbo and Miwa \cite{JM80,JM81}. Also, making use of (\ref{3.13}) in
(\ref{3.17}), and \eqref{highT_G2} with \eqref{lowT_G2} gives the $s \to 0$ boundary conditions
\begin{equation}\label{3.14a}
{\sigma}^{\pm}_0(s) \mathop{\sim}\limits_{s \to 0} - {1 \over 4} \pm (1 + L(s)){s \over 2}, \qquad
L(s)  := \log \Big( \frac{s}{8} \Big) + \gamma_E ,
\end{equation}
and 
\begin{equation}\label{3.14b}
{\sigma}^{\pm}_2(s) \mathop{\sim}\limits_{s \to 0} \mp\frac{s}{24}.
\end{equation}

The boundary conditions (\ref{3.14a}) and (\ref{3.14b}) suggest seeking series solutions of
(\ref{eq:pV}) and (\ref{Pode-2}) of the form
\begin{equation}\label{3.14c}
{\sigma}^{\pm}_0(s) = \sum_{n=0}^{\infty} \sum_{m=0}^n c^{\pm}_{m,n} L(s)^m s^n, \qquad
{\sigma}^{\pm}_2(s) = \sum_{n=0}^\infty \sum_{m=0}^n k^{\pm}_{m,n} L(s)^m s^n,
\end{equation}
subject to the initial conditions
\begin{equation}\label{3.14d}
c^{\pm}_{0,0} = -{1 \over 4}, \quad c^{\pm}_{1,1} = \pm {1 \over 2}, \quad k^{\pm}_{0,0} = 0, \quad k^{\pm}_{0,1} = \mp\frac{1}{24}.
\end{equation}
Substituting these forms gives recurrences for the unknown coefficients, and these are found to have a unique solution, given (\ref{3.14c}). 
The $-1/64$ term in \eqref{3.13} does not appear as part of the initial conditions but we will need it later in a subsequent calculation.
We find, up to order $s^4$,
\begin{equation}
{\sigma}_0^{\pm}(s) = -\frac{1}{4} \pm \Big( \frac{1}{2}+\frac{1}{2}L(s) \Big)s + \Big( \frac{1}{8}-\frac{1}{4}L(s)-\frac{1}{4}L(s)^2 \Big)s^2
                           \pm \Big( \frac{1}{8}L(s)^2 + \frac{1}{8}L(s)^3 \Big)s^3 + \cdots ,
\end{equation}
up to order $s^5$
\begin{equation}
{\sigma}_1^{\pm}(s) = \frac{1}{4}s \pm \frac{1}{2}s^2 + \Big(-\frac{1}{8} - \frac{3}{4}L(s) - \frac{1}{4}L(s)^2 \Big)s^3 
                          \pm \Big( \frac{1}{4}L(s) + \frac{5}{8}L(s)^2 + \frac{1}{4}L(s)^3 \Big)s^4 + \cdots ,
\end{equation}
and, up to order $s^3$
\begin{equation}
{\sigma}_2^{\pm}(s) = \mp\frac{1}{24}s + \Big( \frac{1}{12} + \frac{1}{24}L(s) \Big)s^2 
                          \pm \Big( \frac{3}{8} - \frac{5}{96}L(s) - \frac{1}{32}L(s)^2 \Big)s^3 + \cdots .
\end{equation}
From the first of these it follows from (\ref{3.17}) that up to terms of order $s^4$, and with $\tilde{A} = 2^{-1/4} A$ (recall the definition of $A$ from (\ref{2.12}))
\begin{equation}\label{g01}
  G_0^{\pm}(s) = \frac{\tilde{A}}{s^{1/4}} \Big( 1 + \frac{1}{2}L(s)(\pm s) + \frac{1}{16}s^2 + \frac{1}{32}L(s)(\pm s)^3 + \cdots \Big) ,
\end{equation}
which is consistent with a result presented in \cite[Eq.~(8.10)]{AP01}. 
In  \cite[Eq.~(8.10)]{AP01} the expansion for $G_0^+(s)$ is given up to and including terms of order $s^{10}$. 
Extending our expansion to that order gives agreement; the corresponding expansion
of $G_0^-(s)$ is obtained by simply replacing $s$ by $-s$ in all terms except the factor of $1/s^{1/4}$.
From the second we find that up to terms of order $s^5$,
\begin{equation}\label{g02}
G_1^{\pm}(s) = \frac{\tilde{A}}{s^{1/4}} \Big( \pm \frac{1}{4}s \pm \Big( \frac{1}{2} +\frac{4\pm 1}{8}L(s) \Big)s^2 + \frac{8\pm 1}{64}s^3
               \pm\Big( \frac{1}{32} + \frac{12\pm 1}{128}L(s) \Big)s^4 + \cdots \Big).
\end{equation}
And lastly the initial terms up to order $s^5$, where the integration constant $-1/64$ appears in \eqref{highT_G2} and \eqref{lowT_G2}
\begin{equation}
   G^{+}_2(s) = \frac{A}{(2s)^{1/4}}\left( - \frac{1}{64} - \frac{1}{24}s - \frac{1}{128}sL(s) + \frac{319}{1024}s^2 + \frac{395}{384}s^3 + \frac{1919}{2048}s^3L(s) + \ldots  \right) ,
\label{G2_highT_start}
\end{equation}
and 
\begin{equation}
   G^{-}_2(s) = \frac{A}{(2s)^{1/4}}\left( - \frac{1}{64} + \frac{1}{24}s + \frac{1}{128}sL(s) - \frac{193}{1024}s^2 - \frac{299}{384}s^3 - \frac{895}{2048}s^3L(s) + \ldots \right) .
\label{G2_lowT_start}
\end{equation}
Of course generating these series to (much) higher order is straightforward using computer algebra software. 
In Appendix \ref{AppA} we record the expansions of $G^{(2)\pm}(s)$ up to and including order $s^{15}$.

\section{Comparison with numerical data}\label{S4}
 
With $N$ and $t$ varied so that the scaling variable (\ref{scaling-var}) is fixed, to be consistent with
(\ref{trueGdefn}) and (\ref{2.15}) we must be able to expand
\begin{equation}\label{5.1}
C_{N,N}(t) = {a(s) \over N^{1/4}} + {b(s) \over N^{5/4}} + {c(s) \over N^{9/4}} + \cdots ,
\end{equation}
where
\begin{equation}\label{5.2}
a(s) = (2s)^{1/4} G_0^\pm(s), \qquad b(s) = (2s)^{1/4} G_1^\pm(s), \qquad c(s) = (2s)^{1/4} G_2^\pm(s).
\end{equation}
On the other hand, numerical values of $C_{N,N}(t)$ for particular $N$ and $t$ can readily be
computed from the Toeplitz determinant formula (\ref{2.5}) using the ${}_2 F_1$ form of the entries
(\ref{F1}) and (\ref{F2}). If we fix $s$, then numerical values of the first three coefficients in (\ref{5.1})
can be estimated by truncating (\ref{5.1}) to the first three terms, choosing three distinct (large) $N$
values, and solving for $a(s), b(s), c(s)$. 
In practice the data appearing in Figures \ref{FigG2_verysmall}, \ref{FigG2_small}, \ref{FigG2_large} was interpolated from \eqref{5.1} 
(with an additional term $d(s)/N^{13/4}$) using the values $N=97$, $98$, $99$, $100$ over some fixed interval for $s$.
That the results are accurate is evidenced by their stability
upon repeating this procedure with different choices of three $N$ values. This means, using (\ref{5.2}),
that we have available numerical values against which we can compare our theoretical predictions.

\begin{figure}[h]
    \centering
    \begin{subfigure}[t]{0.5\textwidth}
        \centering
        \includegraphics[height=48mm]{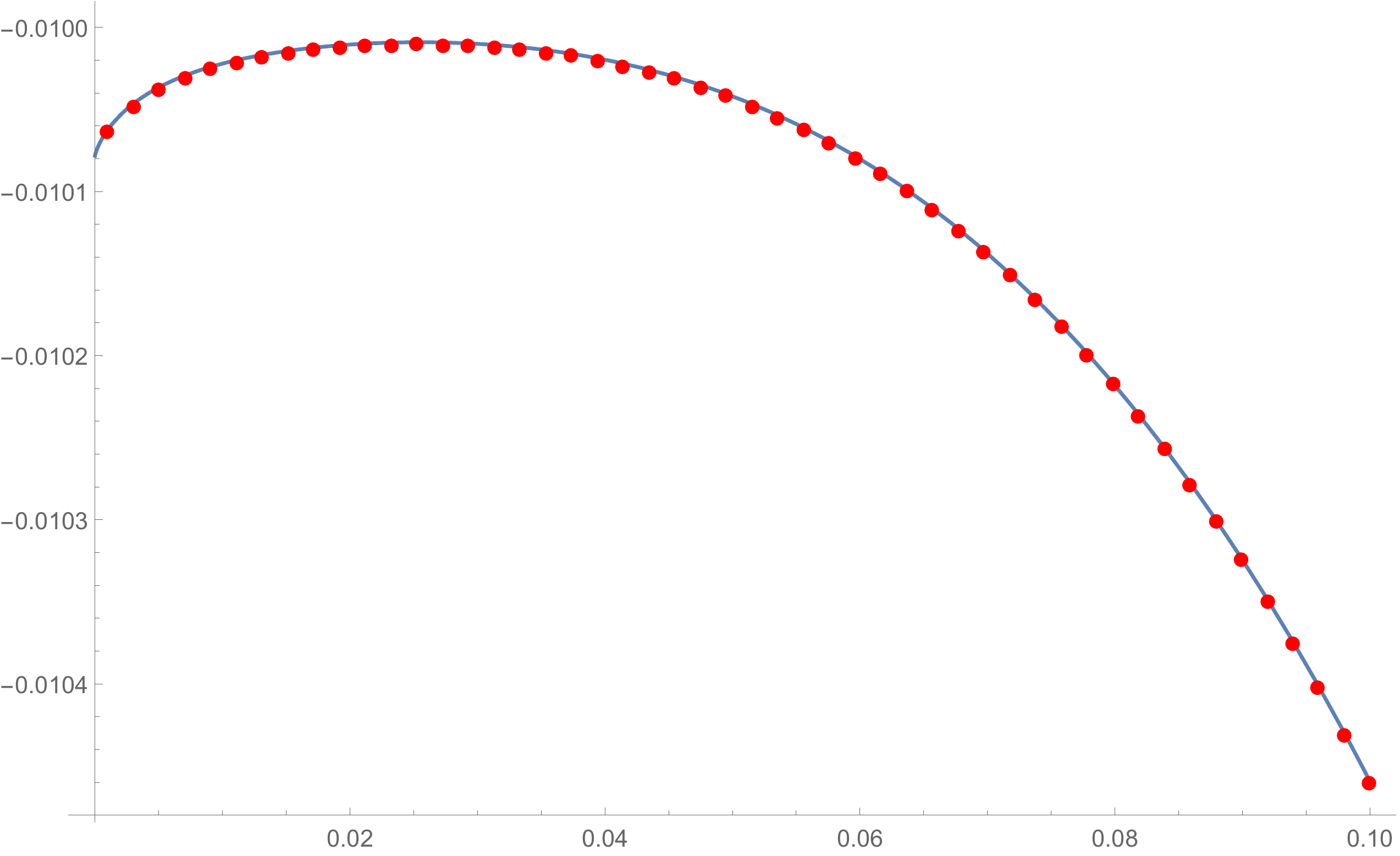}
        \caption{$T > T_C$}
    \end{subfigure}%
    ~ 
    \begin{subfigure}[t]{0.5\textwidth}
        \centering
        \includegraphics[height=48mm]{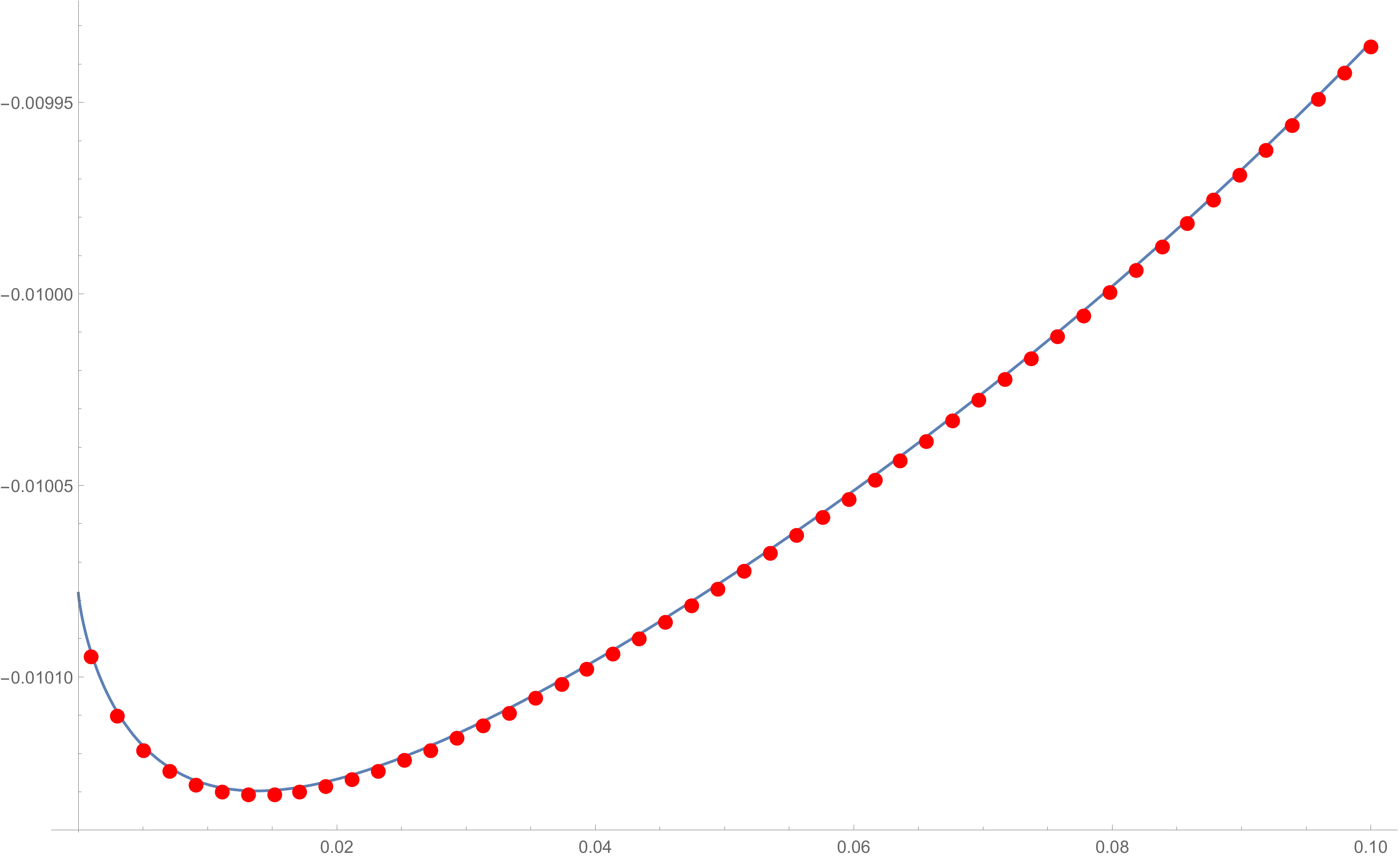}
        \caption{$T < T_C$}
    \end{subfigure}
    \caption{Plots of $(2s)^{1/4}G^+_2(s)$ (left-hand panel) and $(2s)^{1/4}G^-_2(s)$ (right-hand panel) for small $s$ on the interval $[0,0.1]$. 
    Numerical results are the points (red). The series solution (blue) is up to order $s^{20}$.}
    \label{FigG2_verysmall}
\end{figure}

\begin{figure}[h]
    \centering
    \begin{subfigure}[t]{0.5\textwidth}
        \centering
        \includegraphics[height=48mm]{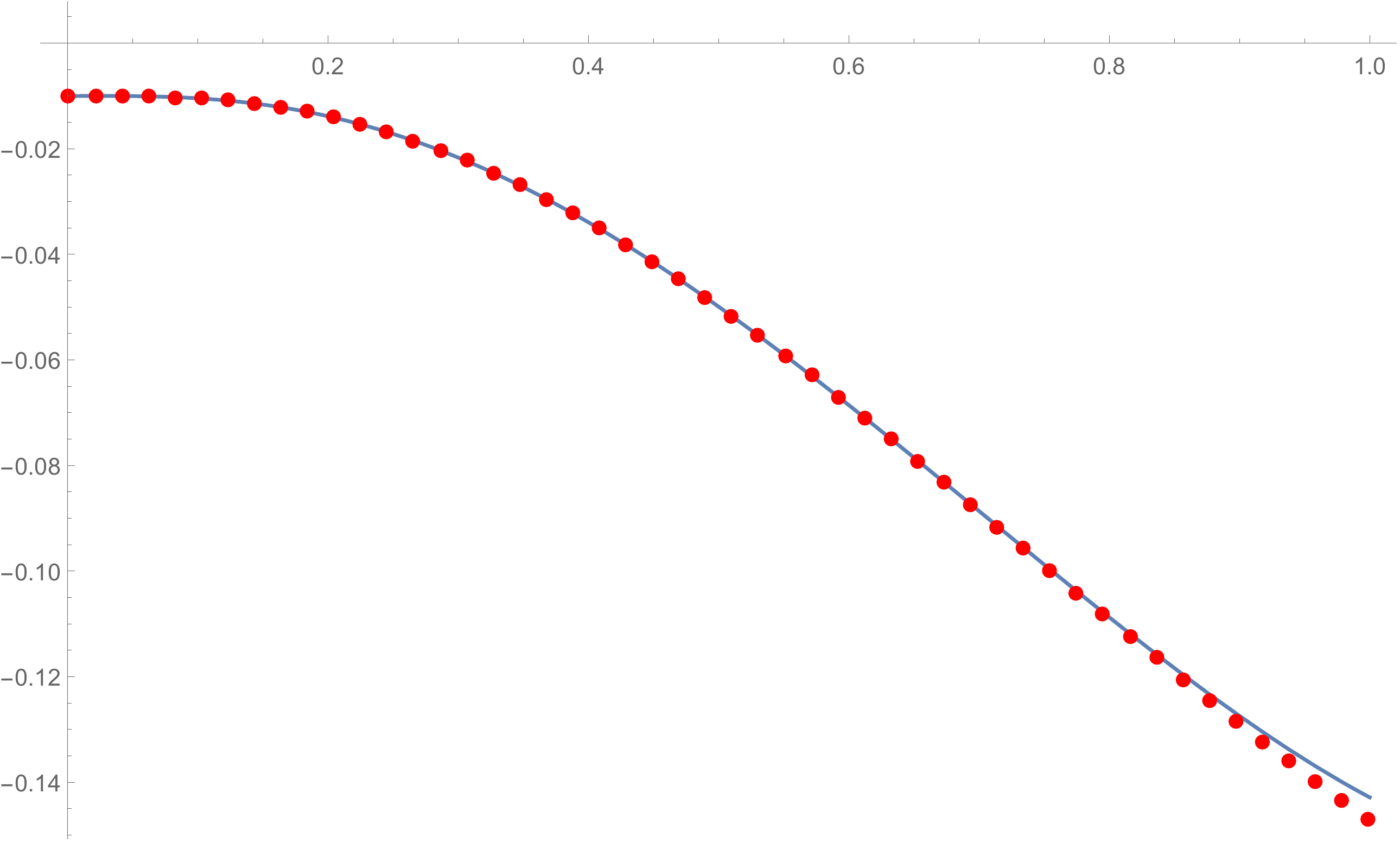}
        \caption{$T > T_C$}
    \end{subfigure}%
    ~ 
    \begin{subfigure}[t]{0.5\textwidth}
        \centering
        \includegraphics[height=48mm]{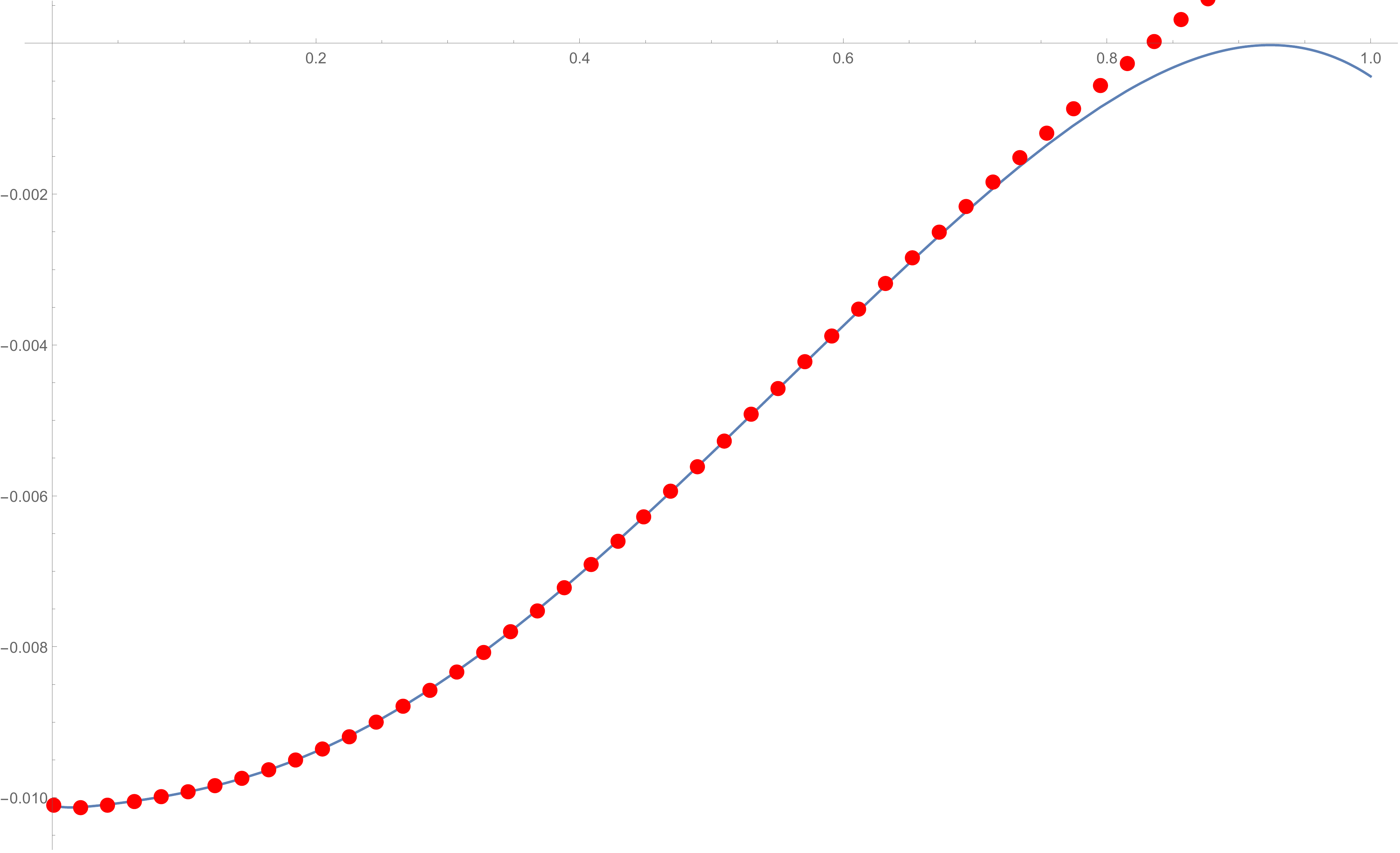}
        \caption{$T < T_C$}
    \end{subfigure}
    \caption{Plots of $(2s)^{1/4}G^+_2(s)$ (left-hand panel) and $(2s)^{1/4}G^-_2(s)$ (right-hand panel) for small $s$ on the interval $[0,1]$. 
    Numerical results are the points (red). The series solution (blue) is up to order $s^{20}$.}
    \label{FigG2_small}
\end{figure}

\begin{figure}[h]
    \centering
    \begin{subfigure}[t]{0.5\textwidth}
        \centering
        \includegraphics[height=48mm]{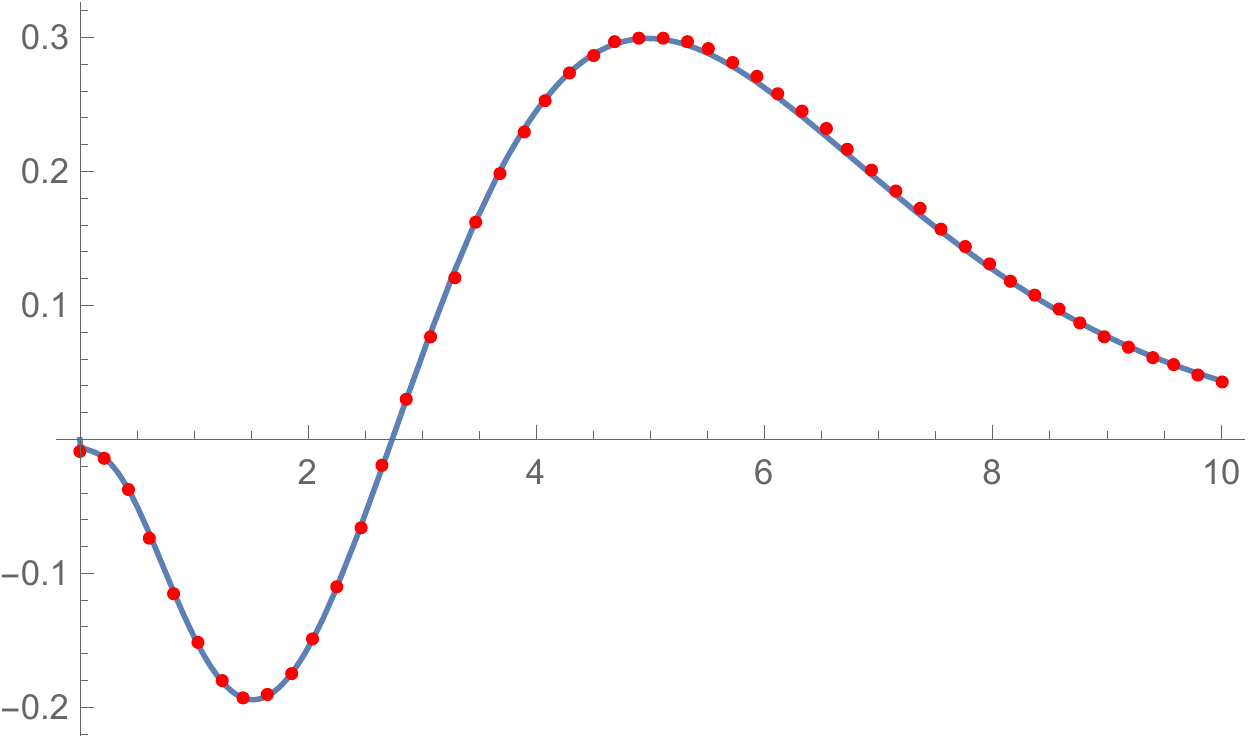}
        \caption{$T > T_C$}
    \end{subfigure}%
    ~ 
    \begin{subfigure}[t]{0.5\textwidth}
        \centering
        \includegraphics[height=48mm]{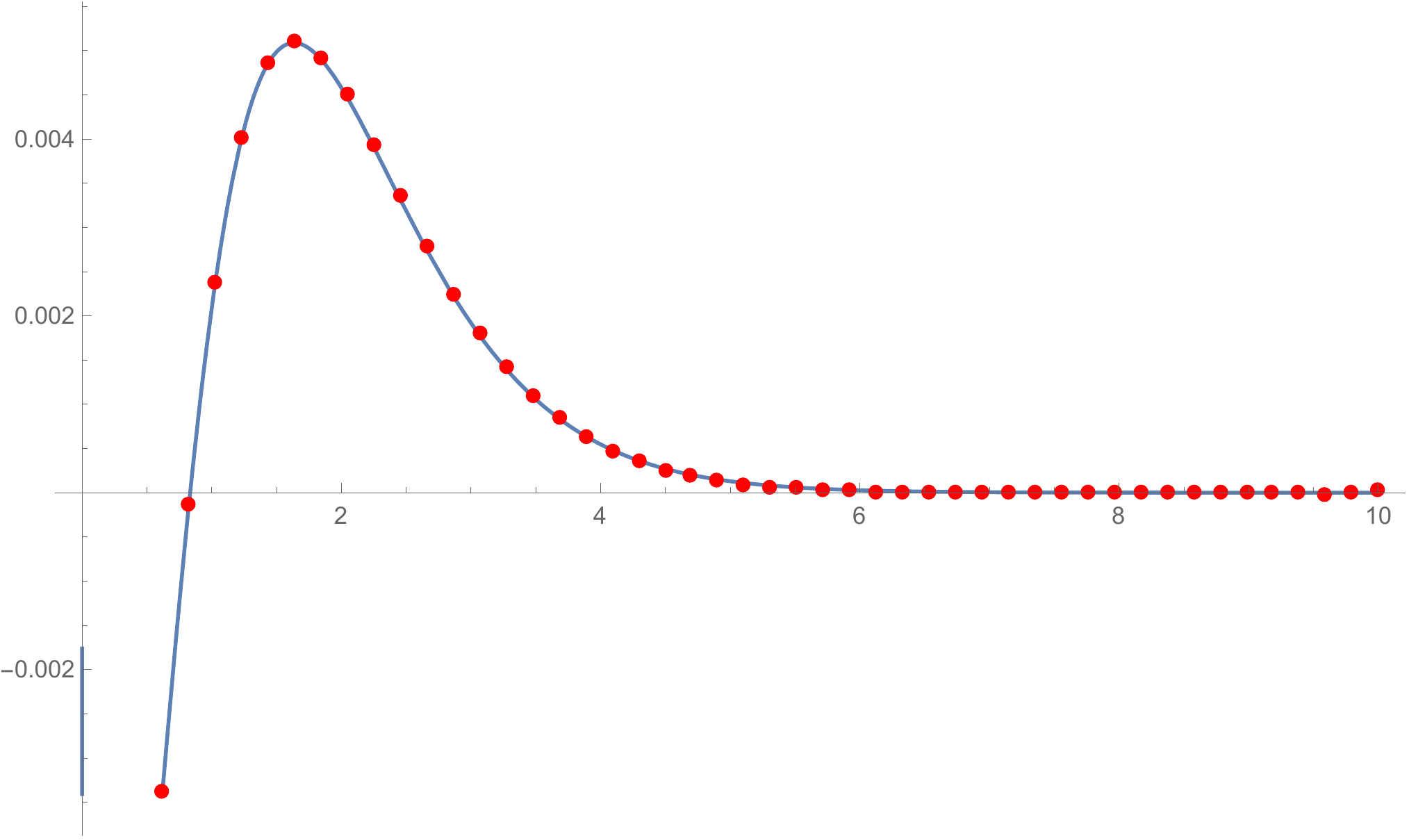}
        \caption{$T < T_C$}
    \end{subfigure}
    \caption{Plots of $(2s)^{1/4}G^+_2(s)$ (left-hand panel) and $(2s)^{1/4}G^-_2(s)$ (right-hand panel) for large $s$. 
    Numerical results are the points (red). The asymptotic solution is the solid line (blue).}
    \label{FigG2_large}
\end{figure}

The large $s$ forms of $G_0^\pm$, $G_1^\pm$ and $G_2^\pm$ are given by (\ref{2.15}), (\ref{in3}) and (\ref{in4}) respectively.
Consider first $G_0^\pm$. It has previously been noticed by Au-Yang and Perk \cite{AP01}
that using (\ref{2.15}) for $G_0^\pm(s)$ in the regime $s > 11.4$ $(G^+)$ and $s > 10.2$ $(G^-)$, 
and the small $s$ power series expansion up to and including terms of order $s^{10}$ otherwise, 
gives remarkable accuracy at the cross-over point of better than $5 \times 10^{-24}$ and $2 \times 10^{-31}$ respectively.
In relation to our numerical data, we find an accurate fit of the large $s$ asymptotic expression for both $G_0^\pm$ for the range $s > 0.4$, and thus well into the small $s$ regime.

The large $s$ asymptotic forms of $G_1^\pm$ and $G_2^\pm$ share the same property as that of $G_0^\pm$ as being valid inside the small $s$ regime. 
To see this, our numerical data for $c(s)$ gives graphical agreement with $(2s)^{1/4}G_2^\pm(s)$ down to as small as $s=0.3$ and $0.1$ 
with a relative error of approximately $1\%$ for high and low temperatures respectively. 
Whereas the small $s$ series expansions of $(2s)^{1/4}G_2^\pm(s)$ (see Figure \ref{FigG2_small}) agree with our data to the fifth decimal for values $0<s<0.6$. 
It can be now observed that there is region of overlap between the large $s$ asymptotics and the small $s$ series expansion of $G_2^\pm$. 
Combining these results would therefore give the overall picture to the correction term of the correlation function in the scaling limit.
	
\section{Concluding Remarks}\label{S5}
Thus far we have presented four types of evidence for the first non-trivial, large $N$ correction to the diagonal correlations of the square lattice Ising model under 
scaling at criticality and furthermore given precise characterisations of this correction. 
We wish to emphasise that each type is actually founded on completely independent logic from the others yet are consistent with each other.
To summarise:
\begin{enumerate}[label=(\alph*)]
\item\label{summary-1} 
  From the PVI $\sigma$-form \eqref{eq:pvi} we have shown the existence of $ N^{-1}, N^{-2} $ 
  corrections which satisfy the linear, inhomogeneous second-order differential equation \eqref{Pode-1}, \eqref{Pode-2} along with \eqref{eq:D(s)}.
\item\label{summary-2}
  From the form factor expansion of the diagonal correlations given by \eqref{2.12}, \eqref{2.13} and \eqref{scaling-var} (due to Lyberg and McCoy \cite{LM07})
  we have performed a large $ N $ expansion combined with the scaling towards criticality. 
  This has yielded a large $s$ expansion in the zeroth order and the first two orders of correction, with exponential suppression of contributions from higher form factors. 
  Consequently one only needs $ f^{(1)} $ for $ T>T_C $ and $ f^{(2)} $ for $ T<T_C $ to obtain the leading orders and we have given explicit evaluations of these.
\item\label{summary-3}
  Employing the expansion for the diagonal correlation for finite $ N $ about $ t=1 $, or what amounts to an expansion of a particular $\tau$-function 
  of the PVI system with half-integer monodromy and confluent logarithmic monodromy data, given in \eqref{t=1_expansion} we deduced small $s$ expansions
  for the zeroth order and first two corrections in \eqref{3.13}. This result furnished more than just the leading term in the small $s$ expansion, and clearly
  indicates a more interesting exact functional form for the correction. This result was derived in \cite{ONGP01} using the Toeplitz determinant form for
  the correlations and the explicit evaluation of their elements as Gauss hypergeometric functions.
\item\label{summary-4}
  Finally we have made first-principles computations of the Toeplitz determinant form for the correlation with large $ N $ and have accurately entered the 
  scaling regime by performing a $N$ extrapolation on the basis of the known asymptotic dependence of the correlations on $ N $ (see \eqref{5.1}). 
  These numerical computations have vindicated the solutions of the differential equation in \ref{summary-1} and both the large $s$ and small $s$ expansions
  in \ref{summary-2} and \ref{summary-3} respectively.      
\end{enumerate}

\section*{Acknowledgements} 
This research project is part of the program of study supported by the ARC Centre of Excellence for Mathematical \& Statistical Frontiers.
We are grateful for the insights and comments by Barry McCoy.

\appendix
\section{Appendix A}\label{AppA}
According to \eqref{3.17}, \eqref{highT_G2} and \eqref{lowT_G2}, knowledge of the power series form of 
${\sigma}_0^\pm(x)$ and ${\sigma}_2^\pm(x)$ allows us to deduce the power series form of $G_0^\pm(s)$ and $G_2^\pm(s)$. 
The former are deduced by substituting the functional forms (\ref{3.14c}) in the coupled differential system of Proposition \ref{P1}.
With the initial conditions (\ref{3.14d}), a triangular system for the unknown coefficients is obtained, allowing the sought power series to be obtained to high order. 
The first few terms are recorded in \eqref{g01}, \eqref{G2_highT_start} and \eqref{G2_lowT_start} above, 
where it was noted that the explicit form of the expansion of $G_0^+(s)$ up to and including terms of order $s^{10}$ is given. 
The power series expansion of $G_0^-(s)$ follows from this by replacing $s$ by $-s$ in all terms except the overall factor of $1/s^{1/4}$. 
In conclusion we express our results for the non-trivial correction as $ G^{(2)+}(s) $ and $ G^{(2)-}(s) $, 
rather than those for $G_2^+(s)$ and $G_2^-(s)$, 
not only because they are more fundamental but because they are simpler, possessing manifest $ s\mapsto -s $ symmetry and have smaller coefficients.
We find in the small $ s $ regime at high temperature
{\small 
\begin{multline}
\frac{(2s)^{1/4}}{A}G^{(2)+}(s) \sim 
-{\frac{1}{64}} + \left( -\frac{1}{24}-{\frac {L}{128}} \right) s - {\frac {35}{3072}}\,{s}^{2} + \left( -{\frac{5}{384}}-{\frac {33\,L}{2048}} \right) {s}^{3}
\\
+ \left( -{\frac{835}{393216}}+{\frac {29\,L}{49152}}+{\frac {65\,{L}^{2}}{16384}} \right) {s}^{4} 
+ \left( -{\frac{397}{786432}}-{\frac {1615\,L}{786432}} \right) {s}^{5} + \left( -{\frac{347}{6291456}-{\frac {131\,L}{786432}}+{\frac {161\,{L}^{2}}{262144}}} \right) {s}^{6}
\\
+ \left(  {\frac{569}{37748736}}-{\frac {525\,L}{4194304}} \right) {s}^{7} 
+ \left(  {\frac{70655}{38654705664}}-{\frac {15145\,L}{805306368}}+{\frac {15283\,{L}^{2}}{402653184}} \right) {s}^{8}
\\
+ \left(  {\frac{481433}{309237645312}}-{\frac {33563\,L}{8589934592}}-{\frac {1381\,{L}^{2}}{1073741824}}+{\frac {385\,{L}^{3}}{536870912}} \right) {s}^{9}
\\
+ \left(  {\frac{452251}{3092376453120}}-{\frac {36059\,L}{38654705664}}+{\frac {9139\,{L}^{2}}{6442450944}} \right) {s}^{10}
\\
+ \left(  {\frac{1080733}{24739011624960}}-{\frac {191519\,L}{6184752906240}}-{\frac {6799\,{L}^{2}}{51539607552}}+{\frac {1763\,{L}^{3}}{25769803776}} \right) {s}^{11}
\\
+ \left(  {\frac{1279915}{237494511599616}-{\frac {612827\,L}{19791209299968}}+{\frac {42065\,{L}^{2}}{1099511627776}}} \right) {s}^{12}
\\
+ \left(  {\frac{6327259}{17099604835172352}}+{\frac {16246021\,L}{7124835347988480}}-{\frac {375181\,{L}^{2}}{59373627899904}}+{\frac {30821\,{L}^{3}}{9895604649984}} \right) {s}^{13}
\\
+ \left(  {\frac{13056079}{79798155897470976}}-{\frac {3808213\,L}{4749890231992320}}+{\frac {130009\,{L}^{2}}{158329674399744}} \right) {s}^{14}
\\
+ \left( -{\frac{4464773}{354658470655426560}+{\frac {5981813\,L}{53198770598313984}}-{\frac {184645\,{L}^{2}}{949978046398464}}+{\frac {14573\,{L}^{3}}{158329674399744}}} \right) {s}^{15} ,
\end{multline}
}
whereas in the low temperature regime
{\small 	
\begin{multline}
\frac{(2s)^{1/4}}{A} G^{(2)-}(s) \sim 
-{\frac{1}{64}} + \left( \frac{1}{24}+{\frac {L}{128}} \right) s - {\frac {35}{3072}}\,{s}^{2} + \left( {\frac{5}{384}}+{\frac {33\,L}{2048}} \right) {s}^{3}
\\
+ \left( -{\frac{835}{393216}}+{\frac {29\,L}{49152}}+{\frac {65\,{L}^{2}}{16384}} \right) {s}^{4}
+ \left(  {\frac{397}{786432}}+{\frac {1615\,L}{786432}} \right) {s}^{5} + \left( -{\frac{347}{6291456}}-{\frac {131\,L}{786432}}+{\frac {161\,{L}^{2}}{262144}} \right) {s}^{6}
\\
+ \left( -{\frac{569}{37748736}}+{\frac {525\,L}{4194304}} \right) {s}^{7} 
+ \left(  {\frac{70655}{38654705664}}-{\frac {15145\,L}{805306368}}+{\frac {15283\,{L}^{2}}{402653184}} \right) {s}^{8}
\\ 
+ \left( -{\frac{481433}{309237645312}}+{\frac {33563\,L}{8589934592}}+{\frac {1381\,{L}^{2}}{1073741824}}-{\frac {385\,{L}^{3}}{536870912}} \right) {s}^{9}
\\ 
+ \left(  {\frac{452251}{3092376453120}}-{\frac {36059\,L}{38654705664}}+{\frac {9139\,{L}^{2}}{6442450944}} \right) {s}^{10}
\\ 
+ \left( -{\frac{1080733}{24739011624960}}+{\frac {191519\,L}{6184752906240}}+{\frac {6799\,{L}^{2}}{51539607552}}-{\frac {1763\,{L}^{3}}{25769803776}} \right) {s}^{11}
\\ 
+ \left(  {\frac{1279915}{237494511599616}}-{\frac {612827\,L}{19791209299968}}+{\frac {42065\,{L}^{2}}{1099511627776}} \right) {s}^{12}
\\ 
+ \left( -{\frac{6327259}{17099604835172352}}-{\frac {16246021\,L}{7124835347988480}}+{\frac {375181\,{L}^{2}}{59373627899904}}-{\frac {30821\,{L}^{3}}{9895604649984}} \right) {s}^{13}
\\ 
+ \left(  {\frac{13056079}{79798155897470976}}-{\frac {3808213\,L}{4749890231992320}}+{\frac {130009\,{L}^{2}}{158329674399744}} \right) {s}^{14}
\\ 
+ \left(  {\frac{4464773}{354658470655426560}}-{\frac {5981813\,L}{53198770598313984}}+{\frac {184645\,{L}^{2}}{949978046398464}}-{\frac {14573\,{L}^{3}}{158329674399744}} \right) {s}^{15} .
\end{multline}	
}
We have abbreviated $ L(s):=L $.

\providecommand{\bysame}{\leavevmode\hbox to3em{\hrulefill}\thinspace}
\providecommand{\MR}{\relax\ifhmode\unskip\space\fi MR }
\providecommand{\MRhref}[2]{%
  \href{http://www.ams.org/mathscinet-getitem?mr=#1}{#2}
}
\providecommand{\href}[2]{#2}

\end{document}